\documentclass[prb,twocolumn,showpacs,amssymb,superscriptaddress]{revtex4-2}
\usepackage{graphicx}
\usepackage{amsmath}
\usepackage{amsfonts}
\usepackage{amsthm}
\usepackage{amssymb}
\usepackage{amsbsy}
\usepackage{wasysym}
\usepackage{bm}
\usepackage{bbm}
\usepackage{mathrsfs}
\usepackage{color}
\usepackage{hyperref}
\usepackage{braket}
\usepackage{soul}
\setcitestyle{super}
\newtheorem{lemma}{Lemma}
\newtheorem{theorem}{Theorem}

\date{\today}
\begin{document}

\title{A multi-player, multi-team nonlocal game for the toric code}

\author{Vir B. Bulchandani}
\affiliation{Princeton Center for Theoretical Science,  Princeton University, Princeton, New Jersey 08544, USA}
\affiliation{Department of Physics, Princeton University, Princeton, New Jersey 08544, USA}

\author{Fiona J. Burnell}
\affiliation{School of Physics and Astronomy, University of Minnesota, Minneapolis, Minnesota 55455, USA}

\author{S. L. Sondhi}
\affiliation{Rudolf Peierls Centre for Theoretical Physics, University of Oxford, Oxford OX1 3PU, United Kingdom}

\begin{abstract}
Nonlocal games yield an unusual perspective on entangled quantum states. The defining property of such games is that a set of players in joint possession of an entangled state can win the game with higher probability than is allowed by classical physics. Here we construct a nonlocal game that can be won with certainty by $2N$ players if they have access to the ground state of the toric code on as many qubits. By contrast, the game cannot be won by classical players more than half the time in the large $N$ limit. Our game differs from previous examples because it arranges the players on a lattice and allows them to carry out quantum operations in teams, whose composition is dynamically specified. This is natural when seeking to characterize the degree of quantumness of non-trivial many-body states, which potentially include states in much more varied phases of matter than the toric code. We present generalizations of the toric code game to states with $\mathbb{Z}_M$ topological order.
\end{abstract}

\maketitle

\section{Introduction}
Quantum pseudo-telepathy refers to the notion that players who are allowed to share entangled quantum states can win certain cooperative games with higher probability than a set of purely classical players\cite{Brassard2005}. Games that exhibit quantum pseudo-telepathy are called nonlocal games; some well-known examples are the magic square game\cite{MerminMagic,PeresMagic} and the parity game\cite{mermin1990extreme,brassard2005recasting}. Optimal quantum strategies for nonlocal games typically involve a specific choice of quantum state tailored to the rules of the game. Conversely, if quantum players are able to win a nonlocal game with high probability and without communicating classically with one another, this can impose strong constraints on the state that they are using to play the game. A prominent recent application of this idea was the demonstration of so-called ``device-independent self-testing'' for the two-singlet state\cite{SelfTest,DeviceIndependent,SelfTestVaz,MagicSelfTesting,Grilo}. 

The advantage that quantum mechanics provides for nonlocal games ultimately arises from the feature of entangled quantum states known as     contextuality\cite{Spekkens_2008,Abramsky_2011}, or the impossibility of reproducing the statistical ensemble of outcomes of a quantum mechanical experiment from a set of local, classical hidden variables. Thus nonlocal games provide an appealingly intuitive method for capturing the degree of ``quantumness'' of a given physical state. Among such games, the parity game has the unusual property that it is \emph{scalable}\cite{Brassard2005}; for any $N \geq 3$, it is possible to define a parity game on $N$ players that exhibits quantum pseudo-telepathy and can be won with certainty if the players share an $N$-qubit Greenberger-Horne-Zeilinger (GHZ) state before playing the game. This scalability property mimics the scaling behaviour of extensive many-body quantum systems, and therefore suggests a natural point of departure towards condensed matter physics. In particular, the question arises of whether nonlocal games can be used to verify contextuality of the ground states of realistic physical systems. 

This question is of basic theoretical interest, in part because of recent evidence that contextuality is the resource that allows for universal quantum computation\cite{ContextualityMagic}, rather than entanglement\cite{EntUseless}. At the same time, the vast majority of studies of the information-theoretic aspects of many-body quantum states have focused on their entanglement structure rather than their contextuality properties\cite{QuantReview}, with some notable exceptions\cite{Cluster,BellNonloc,DetectNonlocality,BECBell,ColdAtomBell,EnergyNonloc,Bravyi,DanielMiyake,Sheffer_2022}.

Here we report some concrete progress on verifying contextuality in many-body systems, in the form of a nonlocal game that can be won with certainty if the players share the ground state of a toric code Hamiltonian\cite{KITAEV20032}. This game is based on the parity game, but differs from previous examples of nonlocal games in some important respects. For example, in this game the players reside on the bonds of a lattice, and may be organized dynamically in ``teams'' whose spatial arrangement inherits the geometry of the lattice. The multiple possibilities for organizing the players in teams reflect the extensively large number of generators of the stabilizer group of the toric code, and endow the rules of our ``toric code game'' with a similarly large degree of flexibility, which is ultimately a consequence of the model's $\mathbb{Z}_2$ topological order.

In our analysis, we shall focus solely on ground states of the ideal toric code Hamiltonian, rather than the gapped topological \emph{phase} that arises when the toric code Hamiltonian is perturbed\cite{KITAEV20032,LaumannKitaev}. More generally, one might ask how far quantum pseudo-telepathy can arise within entire phases of quantum matter, rather than at isolated points within these phases. We note that the latter question has previously been explored for the $\mathbb{Z}_2\times \mathbb{Z}_2$ symmetry-protected topological phase\cite{DanielMiyake}. In a companion paper\cite{companion}, we study this question systematically for both conventional symmetry-breaking order, as in the ferromagnetic phase of the quantum Ising model, and for more exotic kinds of order, such as the topological order under scrutiny below.

The paper is structured as follows. We begin by describing the rules of the toric code game for $2N$ players arranged on a square lattice, and present a perfect quantum strategy that always wins this game. We next prove a uniqueness theorem for the toric code game, to the effect that a perfect quantum strategy for the toric code game can be used to uniquely determine ground states of the toric code Hamiltonian. The full details of this proof are provided in Appendix \ref{App1}. Finally, we exhibit a generalization of the toric code game to toric code states with $\mathbb{Z}_M$ topological order, based on Boyer's modulo $M$ game\cite{boyer2004extended}.

\section{The toric code game}
\label{sec:TCGame}
The toric code game is a $2N$ player game, with one player for every bond of a square lattice $\mathcal{L}$ and $N \geq 3$. We assume periodic boundary conditions in the vertical and horizontal directions, so that $\mathcal{L}$ lies on the surface of a torus. The game is supervised by a referee, who we will call a ``verifier'' by analogy with the theory of interactive
proof systems in computer science\cite{SelfTestVaz}. From this point of view the players are ``proving'' to the verifier that they are in joint possession of a toric code ground state; we will return to this idea later on.

Before the toric code game is played, any given player does not know in advance whether they will participate in the game. At the beginning of the game, the verifier assigns a subset of the players to $T$ non-intersecting vertical loops of the torus $\{\Gamma_i\}_{i=1}^T$, with one player per bond. They assign another subset of players to each bond of a horizontal dual loop $\widetilde{\Gamma}$, which intersects each of the vertical loops $\{\Gamma_i\}_{i=1}^T$ in a single bond $t_i$. The players on the bonds of a given vertical loop $\Gamma_i$ form a ``team'', and may communicate freely with one another, but may not communicate classically outside their team. Also, the players on the bonds of the dual loop $\widetilde{\Gamma}$ may not communicate classically with one another. See Fig. \ref{Fig1} for one example of an allowed configuration of active players.

\begin{figure}[t]
    \centering
    \includegraphics[width=0.75\linewidth]{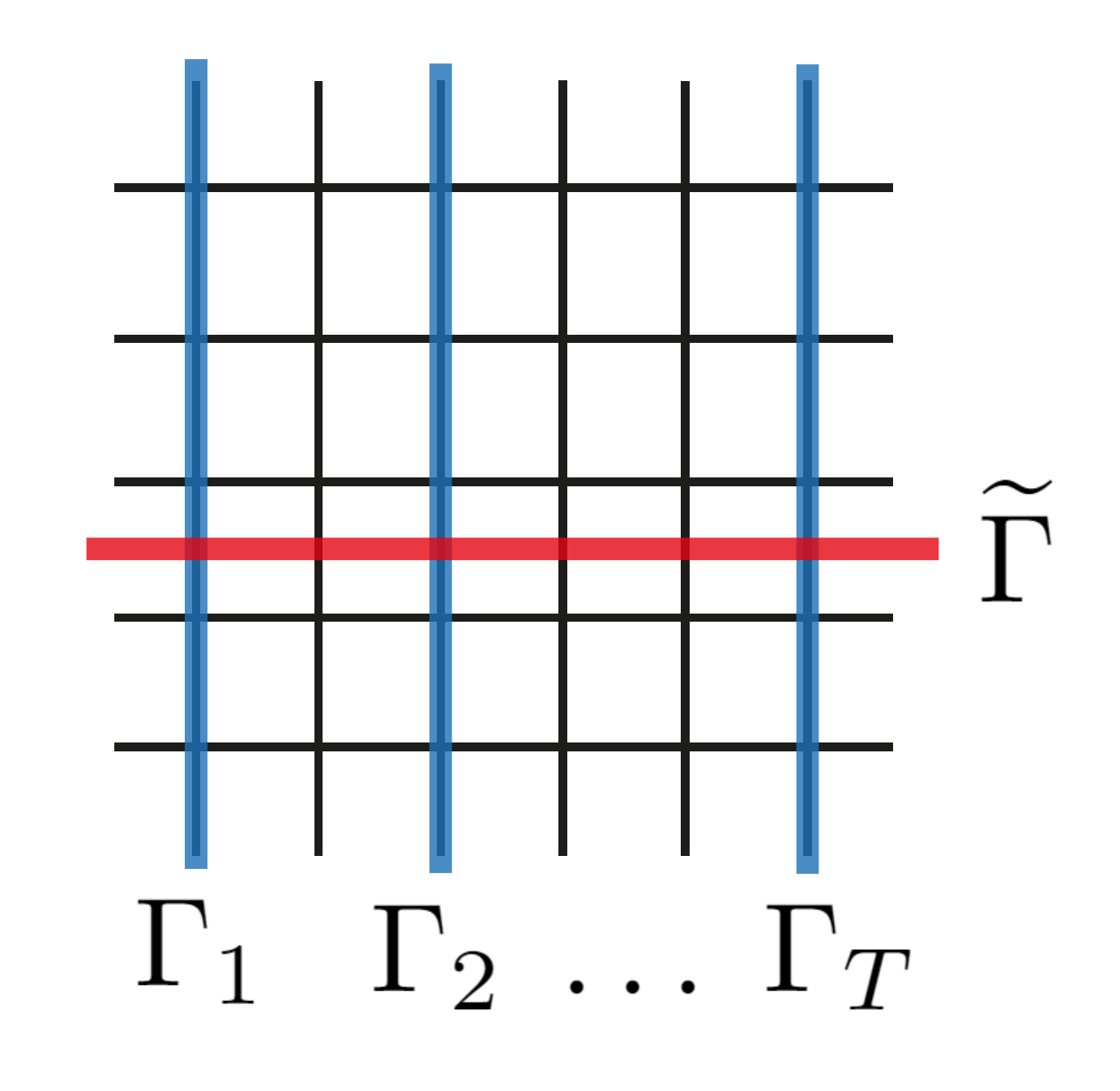}
    \caption{One valid arrangement of players for the toric code game. The vertical loops $\Gamma_1,\, \Gamma_2,\ldots,\Gamma_T$ correspond to ``teams" of players who may communicate classically with one another, but may not communicate between teams. Another subset of players is assigned to the dual loop $\widetilde{\Gamma}$, and there are no doubly occupied bonds, i.e. exactly one player sits on each bond of the union $\left(\cup_{i=1}^T \Gamma_i\right) \cup \widetilde{\Gamma}$. Players on distinct bonds of $\widetilde{\Gamma}$ may not communicate classically with one another.}
    \label{Fig1}
\end{figure}

The verifier gives each team a bit $a_i \in \{0,1\}$, with the promise that $\sum_{i=1}^T a_i$ is even. In order to win the game, each player on a bond of the dual loop $\widetilde{\Gamma}$ must return a bit $y_b$ to the verifier such that
\begin{equation}
\label{eq:wincond}
\sum_{b \in \widetilde{\Gamma}} y_b \equiv \frac{\sum_{i=1}^T a_i}{2} \, \mathrm{mod}\, 2.
\end{equation}

We emphasize that once the underlying square lattice $\mathcal{L}$ is fixed, the verifier may choose the vertical loops $\Gamma_i$, the number of teams $T$ and the dual loop $\widetilde{\Gamma}$ arbitrarily, subject to the constraints described above. For all such choices, there is a perfect quantum strategy that uses the same ground state of the toric code Hamiltonian on $\mathcal{L}$. We now describe this perfect quantum strategy.

We define plaquette and star operators $\hat{A}_p = \prod_{b \in \partial p} \hat{Z}_b$ and  $\hat{B}_s = \prod_{b \in s} \hat{X}_b$ in the usual fashion, in terms of which the toric code Hamiltonian is given by\cite{KITAEV20032,LaumannKitaev}
\begin{equation}
\label{eq:TCH}
\hat{H} = - K \sum_{p} \hat{A}_p - K' \sum_s \hat{B}_s
\end{equation}
where $K,\,K'>0$. The state 
\begin{equation}
|0 0 \rangle = \mathcal{N} \prod_{s} (1+\hat{B}_s)\bigotimes_{b} |\hat{Z}_b = 1\rangle
\end{equation}
where $\mathcal{N}$ is a normalization constant, is a ground state of the toric code Hamiltonian on any closed manifold, since it is an eigenstate with eigenvalue $1$ for any $\hat{A}_p$ or $\hat{B}_s$.  

However, on the torus, this ground state is not unique. To label the full set of degenerate ground states of $\hat{H}$, we define Wilson loop operators $\hat{W}_{x/y} = \prod_{b \in \Gamma_{x/y}} \hat{Z}_b$ and dual Wilson loop operators $\hat{V}_{x/y} = \prod_{b \in \widetilde{\Gamma}_{x/y}} \hat{X}_b$
as depicted in Fig. \ref{Fig2}. The topologically degenerate ground states $|jk\rangle$, $j,k=0,1$ of the model can be labelled by the eigenvalues
\begin{equation}
\hat{W}_x |jk\rangle = (-1)^j |jk\rangle, \quad \hat{W}_y |jk\rangle = (-1)^k |jk\rangle
\end{equation}
of the Wilson loops, and are obtained from the zero-flux ground state by acting with dual Wilson loop operators $|jk\rangle = (\hat{V}_y)^j(\hat{V}_x)^k|00\rangle$. A perfect quantum strategy then proceeds as follows:
\begin{enumerate}
\item Before playing the game, all players share the ``topological cat'' state
\begin{equation}
\label{eq:topcat}
|\psi \rangle = \frac{1}{\sqrt{2}} \left(|00\rangle + |0 1\rangle\right).
\end{equation}
\item Team $i$ acts with a power $\hat{W}_i^{a_i/2}$ of their Wilson loop $\hat{W}_i = \prod_{b \in \Gamma_i} \hat{Z}_b$ on the shared state $|\psi\rangle$.
\item The players on the bonds of the dual loop $\widetilde{\Gamma}$ each measure their qubit in the $\hat{X}$ or Hadamard basis and return its value $y_b \in \{0,1\}$.
\end{enumerate}
To be explicit, the square root of the Wilson loop operator $\hat{W}_i$ in Step 2 is defined by the eigenspace projection $\hat{W}_i^{a_i/2} = \left(\frac{1+\hat{W}_i}{2}\right) + i^{a_i} \left(\frac{1-\hat{W}_i}{2}\right)$. Note that this branch of the square root is non-local, i.e. $\hat{W}_i^{1/2} \neq \prod_{b \in \Gamma_i} \hat{Z}_b^{1/2}$.

Let us now verify that this quantum strategy always wins the toric code game. After Step 2, the players' shared state is given by
\begin{equation}
\label{eq:afterstep2}
|\psi'\rangle = \frac{1}{\sqrt{2}}(|00\rangle + (-1)^{\sum_{i=1}^T a_i/2}|01\rangle).
\end{equation}
Step 3 is equivalent to the players on the dual loop $\widetilde{\Gamma}$ collectively measuring the dual Wilson loop operator $\hat{V}_{\widetilde{\Gamma}} =\prod_{b \in \widetilde{\Gamma}}\hat{X}_b$ on their shared state to yield an eigenvalue $(-1)^{\sum_{b\in\widetilde{\Gamma}} y_b}$. But the shared state after Step 2 is an eigenstate of $\hat{V}_{\widetilde{\Gamma}}$ with eigenvalue $(-1)^{\sum_{i=1}^T a_i/2}$. It follows that
\begin{equation}
(-1)^{\sum_{b\in\widetilde{\Gamma}} y_b} = (-1)^{\sum_{i=1}^T a_i/2},
\end{equation}
which implies the winning condition Eq. \eqref{eq:wincond} and therefore that the players have won the game. Thus we have shown that the topological cat state Eq. \eqref{eq:topcat} yields a perfect quantum strategy for the toric code game for any choice of teams and dual loop consistent with the geometry depicted in Fig. \ref{Fig1}.

To prove that this defines a genuine nonlocal game, we must check that the best possible classical strategy wins with probability less than $1$. In fact, the condition of no classical communication between the players on the bonds of $\widetilde{\Gamma}$ implies\footnote{This follows by analogous arguments for the parity game\cite{brassard2005recasting}.} that the optimal classical strategy is achieved if the players at the intersecting bonds $\{t_i\} = \Gamma_i \cap \widetilde{\Gamma}$ apply an optimal classical strategy for the $T$ qubit parity game, as tabulated by Brassard-Broadbent-Tapp\cite{brassard2005recasting}, while the remaining players on bonds $b \in \widetilde{\Gamma} \backslash \cup_{i=1}^T \{t_i\}$ return $y_b = 0$. Thus the best classical strategy for the toric code game with $T$ teams wins with probability
\begin{equation}
\label{eq:classwinprob}
p^*_{\mathrm{cl}} = \frac{1}{2} + \frac{1}{2^{\lceil T/2\rceil}}
\end{equation}
and the toric code game is a genuine nonlocal game whenever the number of teams $T \geq 3$. Since it is possible for the verifier to assign up to $T= \mathcal{O}(\sqrt{N})$ teams in an instance of the toric code game, the ``most quantum'' instance of the game has classical probability of winning tending to $1/2$ as the number of players tends to infinity,
\begin{equation}
p^{*}_{\mathrm{cl}} = \frac{1}{2} + \frac{1}{2^{\mathcal{O}(\sqrt{N})}} \to \frac{1}{2}, \quad N \to \infty,
\end{equation}
and always defines a nonlocal game provided the lattice $\mathcal{L}$ is more than three qubits wide.

Finally, we note that the version of the toric code game described above explicitly pairs a set of teams with a single generator of the fundamental group of the torus, in the sense that all the teams in Fig. \ref{Fig1} lie along vertical loops. More generally, one can define analogous nonlocal games on surfaces of arbitrary genus, with arbitrary allowed lattice geometries such as honeycomb lattices\cite{StringNet}, and assign a set of teams to any generator of the fundamental group. Furthermore, the game can be played simultaneously by multiple sets of teams who are each assigned to independent generators of the fundamental group. One could also consider \emph{reducing} the number of fundamental group generators and playing the toric code game with a single logical qubit, for example on the surface of a cylinder or on a patch of the surface code with appropriate boundary conditions\cite{SurfaceCodes}.

For each of these generalizations, there is a perfect quantum strategy that uses the appropriate topological cat state. For example, in one simple generalization of the genus $g=1$ square lattice toric code game described above, the verifier assigns players to horizontal teams and a vertical dual loop; then the perfect quantum strategy is a spatial reflection of the strategy described above, and involves the players sharing the topological cat state $|\psi\rangle = (|00\rangle + |10\rangle)/\sqrt{2}$ before playing the game.  Moreover, the operations carried out by these reflected teams commute with the operations of the original teams, such that the two sets of teams can play simultaneously and both sets can win with certainty.
\begin{figure}[t]
    \centering
    \includegraphics[width=0.99\linewidth]{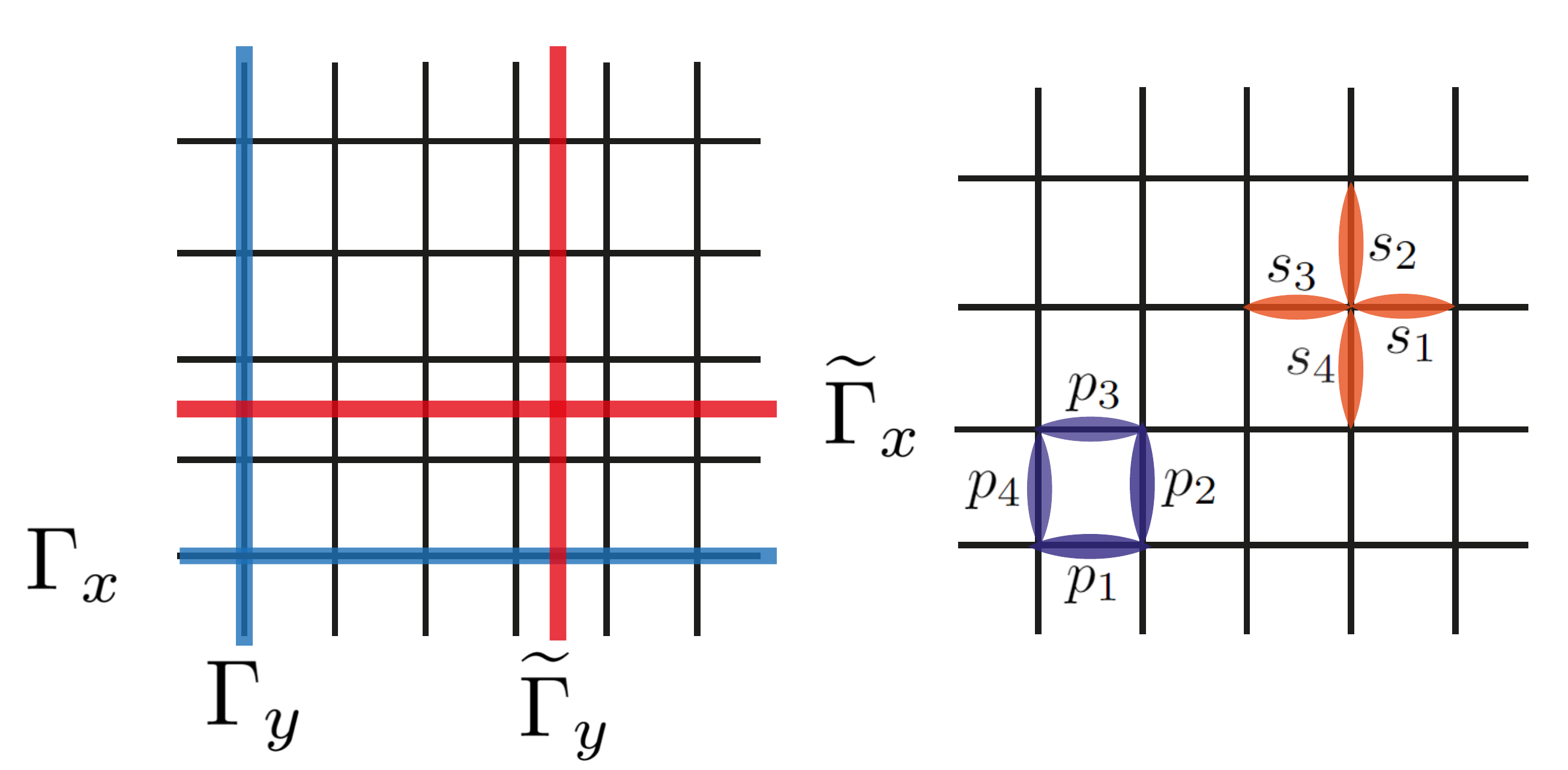}
    \caption{Labelling conventions for Wilson loop operators on the square lattice (\textit{left}) and plaquettes and stars in the $\mathbb{Z}_M$ toric code (\textit{right}).}
    \label{Fig2}
\end{figure}

\section{Uniqueness theorem for the toric code game}
We now prove that a perfect quantum strategy for the toric code game can be used to uniquely determine ground states of the toric code Hamiltonian. We first adapt some definitions from a companion paper\cite{companion} to the present context. We define a \emph{quantum strategy} $\mathcal{S} = (|\psi\rangle,\mathcal{P})$ for an instance of the toric code game to consist of
\begin{enumerate}
\item a $2N$-qubit pure state $|\psi\rangle$, which is shared by all players before the game begins.
\item a \emph{protocol} $\mathcal{P}$, which is a set of quantum gates and measurements that each team or player applies to their qubits.
\end{enumerate}
We define a \emph{perfect quantum strategy} for an instance of the toric code game to be any quantum strategy $\mathcal{S}$ that always wins that game for all allowed inputs\cite{brassard2005recasting}. For concreteness, we will follow the conventions for the toric code game defined in Section \ref{sec:TCGame} and in Fig. \ref{Fig1}, and let $\mathcal{P}_{\mathrm{TC}}$ denote the quantum protocol described in that section. By an ``instance'' of the toric code game, we mean a specific choice of the number of teams $T\geq 3$, vertical loops $\{\Gamma_{j}\}_{j=1}^T$ and a horizontal dual loop $\widetilde{\Gamma}$.

Our uniqueness theorem relies upon the following Lemma, whose proof we defer to Appendix \ref{App1}:
\begin{lemma}
\label{Lemma1}
Let $|\psi\rangle$ be a $2N$-qubit pure state. Then for a given instance of the toric code game, the quantum strategy $\mathcal{S} = (|\psi\rangle,\mathcal{P}_{\mathrm{TC}})$ wins the toric code game with probability
\begin{align}
\nonumber &p_{\mathrm{qu}}(|\psi\rangle)-1/2 \\
= &\sum_{\{\vec{\sigma} \in \{0,1\}^{2N}: \hat{W}_j|\vec{\sigma}\rangle = |\vec{\sigma}\rangle\}} \frac{1}{2}(|\langle \psi |\varphi_{\widetilde{\Gamma}}^{+}(\vec{\sigma})\rangle|^2 - |\langle \psi |\varphi^{-}_{\widetilde{\Gamma}}(\vec{\sigma})\rangle|^2),
\end{align}
where the summation is over all computational basis states $|\vec{\sigma}\rangle =\bigotimes_{b} |\sigma_b\rangle$ that satisfy $\hat{W}_j |\vec{\sigma}\rangle = |\vec{\sigma}\rangle$ for $j=1,2,\ldots,T$, and the states
\begin{align}
\label{WilsonCat}
|\varphi^{\pm}_{\widetilde{\Gamma}}(\vec{\sigma})\rangle = \frac{1}{\sqrt{2}}(1 \pm \hat{V}_{\widetilde{\Gamma}})|\vec{\sigma}\rangle.
\end{align}
\end{lemma}
This Lemma allows us to prove the following result.
\begin{theorem}
\label{thm1}
Let $\mathcal{S} = (\mathcal{|\psi\rangle},\mathcal{P}_{\mathrm{TC}})$ be a perfect quantum strategy for all instances of the toric code game. Then $|\psi\rangle$ is a ground state of the toric code with
\begin{equation}
|\psi\rangle \in \mathrm{span}\left\{\frac{1}{\sqrt{2}}(|00\rangle+|01\rangle),\frac{1}{\sqrt{2}}(|10\rangle+|11\rangle)\right\}.
\end{equation}
\end{theorem}  
\begin{proof}
Suppose $\mathcal{S} = (|\psi\rangle,\mathcal{P}_{\mathrm{TC}})$ is a perfect quantum strategy for all instances of the toric code game. We assume that the state $|\psi\rangle$ is unit normalized. By Lemma \ref{Lemma1}, it follows that
\begin{align}
\nonumber |\psi\rangle \in \bigcap_{(T, \{\Gamma_j\}_{j=1}^T,\widetilde{\Gamma})} &\mathrm{span} \bigg\{\bigg\{\frac{1}{\sqrt{2}}(1+\hat{V}_{\widetilde{\Gamma}})|\vec{\sigma}\rangle : \vec{\sigma}\in \{0,1\}^{2N}, \\
\, &\hat{W}_j|\vec{\sigma}\rangle= |\vec{\sigma}\rangle, \,j=1,2,\ldots,T\bigg\}\bigg\}.
\end{align}
Thus for any horizontal dual Wilson loop $\hat{V}_x$, it must be true that
\begin{equation}
\label{eq:phi}
|\psi\rangle = \frac{1}{\sqrt{2}}(1+\hat{V}_x)|\phi\rangle,
\end{equation}
where $|\phi\rangle$ is some state for which
\begin{equation}
\label{eq:Wilsony}
\hat{W}_y |\phi\rangle = |\phi\rangle
\end{equation}
for any vertical Wilson loop $\hat{W}_y$ that intersects $\hat{V}_x$ in a single bond. Note that unit normalization of $|\psi\rangle$ implies
\begin{equation}
1 = \langle \psi | \psi \rangle = \langle \phi | (1+\hat{V}_x)| \phi \rangle = \langle \phi | \phi \rangle,
\end{equation}
since $\hat{V}_x$ changes the $\hat{W}_y$ quantum number of $|\phi\rangle$. Now let $\hat{B}_s$ be any star operator and choose horizontal dual loops $\hat{V}_1$ and $\hat{V}_2$ such that $\hat{V}_1\hat{V}_2=\hat{B}_s$. It follows that
\begin{equation}
|\psi\rangle = \frac{1}{\sqrt{2}}(1+\hat{V}_1)|\phi_1\rangle = \frac{1}{\sqrt{2}}(1+\hat{V}_2)|\phi_2\rangle
\end{equation}
for some $|\phi_1\rangle, \, |\phi_2\rangle$ on which any allowed $\hat{W}_y=1$, with $\langle \phi_1 | \phi_1 \rangle = \langle \phi_2 | \phi_2 \rangle = 1$. Normalization of $|\psi\rangle$ can be written as
\begin{equation}
\frac{1}{2}\langle \phi_1| 1 + \hat{V}_1 + \hat{V}_2 + \hat{V}_1 \hat{V_2} | \phi_2 \rangle = 1.
\end{equation}
Since $\hat{V}_{1}$ and $\hat{V}_2$ individually change the $\hat{W}_y$ quantum number of $|\phi_2\rangle$, we deduce that
\begin{equation}
\langle \phi_1 | \phi_2 \rangle + \langle \phi_1| \hat{B}_s | \phi_2 \rangle = 2.
\end{equation}
But this equality requires both terms to attain their maximum possible values consistent with the Cauchy-Schwartz inequality, so that $\langle \phi_1 | \phi_2 \rangle = 1$ and $\langle \phi_1 | \hat{B}_s | \phi_2 \rangle = 1$. The first relation implies that
\begin{equation}
|\phi_2\rangle = |\phi_1\rangle,
\end{equation}
proving that the state $|\phi\rangle$ in Eq. \eqref{eq:phi} is independent of the choice of horizontal dual loop $\hat{V}$, while the second relation implies that
\begin{equation}
\label{eq:star}
\hat{B}_s|\phi\rangle = |\phi\rangle.
\end{equation}
We next let $\hat{A}_p$ be any plaquette operator and choose vertical Wilson loops $\hat{W}_1$ and $\hat{W}_2$ such that $\hat{W}_1 \hat{W}_2 = \hat{A}_p$. Note that it is always possible to choose $\hat{W}_1$, $\hat{W}_2$ and a dual loop $\hat{V}_x$ that intersects both $\hat{W}_1$ and $\hat{W}_2$ in a single bond, so that $|\psi\rangle = \frac{1}{\sqrt{2}}(1+\hat{V}_x)|\phi\rangle$ with $\hat{W}_1 |\phi\rangle = \hat{W}_2 |\phi\rangle = |\phi\rangle$. It follows that
\begin{equation}
\label{eq:plaq}
\hat{A}_p |\phi\rangle = \hat{W}_1\hat{W}_2|\phi\rangle = |\phi\rangle.
\end{equation}
Equations Eq. \eqref{eq:star} and \eqref{eq:plaq} guarantee that $|\phi\rangle$ is locally a ground state of the toric code Hamiltonian. Equation Eq. \eqref{eq:Wilsony} then fixes the quantum number of vertical Wilson loops, so that
\begin{equation}
|\phi\rangle \in \mathrm{span}\{|00\rangle,|10\rangle\}.
\end{equation}
The result follows by Eq. \eqref{eq:phi}.
\end{proof}

 We note that an analogous result holds for the doubled semion state\cite{StringNet}. To be precise, just as a quantum winning strategy using the Wilson loops of the toric code serves to determine toric code ground states, the analogous quantum winning strategy using Wilson loops of the doubled semion model can be used to determine doubled semion ground  states.

\section{$\mathbb{Z}_M$ toric code games}

In this section, we introduce nonlocal games that can be won with certainty if the players share a ground state of the toric code with $\mathbb{Z}_M$ topological order ($M > 2$) before playing the game. These games are based on Boyer's modulo $M$ divisor $M$ generalizations of the parity game\cite{boyer2004extended}. At the level of assigning players to bonds and teams, the rules of this ``$\mathbb{Z}_M$ toric code game'' are identical to those of the toric code game described in Section \ref{sec:TCGame}. However, in the $\mathbb{Z}_M$ toric code game, the verifier provides each team with a number $a_i \in \{0,1,\ldots,M-1\}$, rather than a single bit. 
Similarly, to win the game, the players at each bond of $\widetilde{\Gamma}$ must return numbers $y_b \in \{0,1,\ldots,M-1\}$ to the verifier such that 
\begin{equation}
\label{eq:zmwincond}
\sum_{b \in \widetilde{\Gamma}} y_b \equiv \frac{\sum_{i=1}^T a_i}{M} \,\mathrm{mod}\,{M}.
\end{equation}
An example of a valid arrangement of players is depicted in Fig. \ref{Fig1}. This recovers the toric code game in the case $M=2$.

To define the $\mathbb{Z}_M$ toric code Hamiltonian, let $\omega_{M} = e^{i2\pi/M}$ denote an $M$th root of unity and define $\mathbb{Z}_M$ ``clock'' and ``shift'' operators 
\begin{align}
\hat{C} = \begin{pmatrix} 
1 & 0 & \ldots & 0 \\
0 & \omega_M & \ldots & 0 \\
0 & 0 & \ddots & 0 \\
0 & 0 & \ldots & \omega_M^{M-1}
\end{pmatrix},\quad \hat{S} = \begin{pmatrix} 
0 & 0 & \ldots & 1 \\
1 & 0 & \ldots & 0 \\
0 & 1 & \ddots & 0 \\
0 & 0 & \ldots & 0
\end{pmatrix},
\end{align}
which satisfy the commutation relations $\hat{C}\hat{S} = \omega_M \hat{S}\hat{C}$

We next label plaquettes and stars as shown in Fig. \ref{Fig2} and define plaquette and star operators $\hat{A}_p = \frac{1}{2}(\hat{C}_{p_1}^\dagger \hat{C}_{p_2}^\dagger \hat{C}_{p_3} \hat{C}_{p_4} + \mathrm{h.c.})$ and $\hat{B}_s = \frac{1}{2}(\hat{S}_{s_1}^\dagger \hat{S}_{s_2}^\dagger \hat{S}_{s_3} \hat{S}_{s_4} + \mathrm{h.c.})$, which mutually commute thanks to the labelling convention. In terms of these operators, the Hamiltonian of the $\mathbb{Z}_M$ toric code is given by\cite{KITAEV20032}
\begin{equation}
\hat{H} = - K \sum_{p}\frac{1}{M} \sum_{j=0}^{M-1} (\hat{A}_p)^j - K' \sum_s \frac{1}{M}\sum_{j=0}^{M-1} (\hat{B}_s)^j
\end{equation}
with $K,K'>0$.

The zero-flux ground state can be obtained by projecting a $\hat{C}_b=1$ product state into a $\hat{B}_s=1$ state, and we can write it as
\begin{equation}
|00\rangle = \mathcal{N} \prod_s \frac{1}{M} \sum_{j=0}^{M-1} (\hat{B}_s)^j \bigotimes_{b=1}^N |\hat{C}_b = 1\rangle.
\end{equation}

To label all $M^2$ ground states of $\hat{H}$, we introduce ``clock'' Wilson loop operators $
\hat{W}_{x/y} = \prod_{b \in \Gamma_{x/y}} \hat{C}_b$ and their dual ``shift'' Wilson loop operators $
\hat{V}_{x/y} = \prod_{b \in \widetilde{\Gamma}_{x/y}} \hat{S}_b$,
with labelling conventions for loops and dual loops as in Fig. \ref{Fig2}. The degenerate ground states $|jk\rangle$ can be labelled by the eigenvalues
\begin{equation}
\hat{W}_x |jk\rangle = \omega_M^j |jk\rangle, \quad \hat{W}_y |jk\rangle = \omega_M^k |jk\rangle
\end{equation}
of the clock loops, and are obtained from the zero-flux ground state by repeated action of the shift loops, namely $|jk\rangle = (\hat{V}_y)^j(\hat{V}_x)^k|00\rangle$.

Then the following quantum strategy always wins:
\begin{enumerate}
    \item Before playing the game, all players share the state
    \begin{equation}
    |\psi \rangle = \frac{1}{\sqrt{M}} \left(|00\rangle + |01\rangle + \ldots |0M-1 \rangle\right).
    \end{equation}
    \item Team $i$ acts with an $M$th root $\hat{W}_i^{-a_i/M}$ of their clock Wilson loop $\hat{W}_i = \prod_{b \in \Gamma_i} \hat{C}_b$.
    \item The players on the bonds of the dual loop $\widetilde{\Gamma}$ each measure their qudit in the shift basis and return its value $y_b$.
\end{enumerate}
Let us briefly verify that this always yields a perfect quantum strategy for the $\mathbb{Z}_M$ toric code game. After Step 2, the new state is
\begin{equation}
|\psi'\rangle  = \frac{1}{\sqrt{M}} \sum_{k=0}^{M-1}  \omega_{M}^{-k\sum_{i=1}^T a_i/M} |0k\rangle 
\end{equation}
by definition of the ground states $|jk\rangle$. This is an eigenstate of the shift Wilson loop $\hat{V}_{\widetilde{\Gamma}} = \prod_{b\in\widetilde{\Gamma}} \hat{S}_b$ with eigenvalue $\omega_{M}^{\sum_{i=1}^T a_i/M}$. Step 3 yields numbers $y_b \in \{0,1,\ldots,M-1\}$ with the property that $\omega_M^{\sum_{b\in\widetilde{\Gamma}} y_b}$ is the eigenvalue of the shift Wilson loop $\hat{V}_{\widetilde{\Gamma}}$ acting on the final state. But since the state $|\psi'\rangle$ was already an eigenstate of that shift Wilson loop, it follows that
\begin{equation}
\omega_M^{\sum_{b\in\widetilde{\Gamma}} y_b}=\omega_{M}^{\sum_{i=1}^T a_i/M},
\end{equation}
which implies the winning condition Eq. \eqref{eq:wincond}.

Our observation in Section \ref{sec:TCGame} that the optimal classical strategy on the intersecting bonds $t_i$ induces an optimal classical strategy for the toric code game applies equally to the $\mathbb{Z}_M$ case, and demonstrates that the $\mathbb{Z}_M$ toric code game is a nonlocal game for a lattice $\mathcal{L}$ that is more than three qubits wide. This is because the $T$ qudit modulo $M$ divisor $M$ game is a nonlocal game\cite{boyer2004extended} for $T \geq 3$ qudits. For the ``most quantum'' instance of the $\mathbb{Z}_M$ toric code game with $T = \mathcal{O}(\sqrt{N})$ teams, it also follows\cite{companion} that the probability of winning for the optimal classical strategy
\begin{equation}
p^*_{\mathrm{cl}} \to 1/M, \quad N \to \infty,
\end{equation}
with corrections that are exponentially small in $\sqrt{N}$.

Similarly, the $\mathbb{Z}_M$ toric code game has all the freedom of the toric code game described in section \ref{sec:TCGame}: it can be played on any surface of genus greater than one, on any lattice geometry, with teams assigned according to any generator of the fundamental group, or even with multiple sets of teams that are associated with different fundamental group generators and play simultaneously.

\section{Conclusion}

We have proposed a nonlocal game that can be won with certainty if the players share ground states of the toric code before playing the game. This game has several features that distinguish it from previous examples of nonlocal games, including the spatial arrangement of the players on a lattice, the dynamical organization of players into ``teams'' and the topological invariance of quantum winning strategies under continuous deformations of these teams. We have further shown that a perfect quantum strategy for the toric code game can be used to uniquely determine toric code states (see Theorem \ref{thm1}), and exhibited generalizations of the toric code game to toric code states with $\mathbb{Z}_M$ topological order.

The former result can be seen as a first step towards device-independent self-testing\cite{SelfTest,DeviceIndependent,SelfTestVaz,MagicSelfTesting,Grilo} or certification of non-trivial ground states of condensed matter Hamiltonians. A desirable goal would be to extend Theorem \ref{thm1} to a full ``rigidity'' result of the type proved for the two-singlet state\cite{MagicSelfTesting}, whereby the rules of the toric code game alone could suffice to determine the toric code ground state, up to local unitary operations. We emphasize that the latter ambiguity cannot be removed; for example, if a set of quantum players share a toric code eigenstate containing some number of electric and magnetic excitations, they can always win the toric code game with certainty by applying suitably chosen unitaries\cite{LaumannKitaev} at their bonds to remove these excitations, before applying the perfect quantum strategy described in Section \ref{sec:TCGame}.

More broadly, the possibility of using quantum games to certify ground states of realistic physical systems up to local unitary operations suggests a novel information-theoretic approach to the problem of classifying phases of matter. A natural question related to this idea is whether games can be used to characterize states with non-Abelian global symmetries, such as Pfaffian quantum Hall states, beyond the Abelian examples studied in this paper.

\section{Acknowledgments}
We thank I. Arad, D. S. Borgnia, A. B. Grilo, R. M. Nandkishore, and especially U. V. Vazirani for helpful discussions. V. B. B. is supported by a fellowship at the Princeton Center for Theoretical Science. F. J. B. is supported by NSF DMR-1928166, and is grateful to the Carnegie Corporation of New York and the Institute for Advanced Study, where part of this work was carried out. This work was supported by a Leverhulme Trust International Professorship grant number LIP-202-014 (S. L. S). For the purpose of Open Access, the author has applied a CC BY public copyright licence to any Author Accepted Manuscript version arising from this submission.

\bibliography{games.bib}

\onecolumngrid
\appendix
\section{Proof of Lemma \ref{Lemma1}}
\label{App1}
In this Appendix, we prove Lemma \ref{Lemma1} of the main text.
\begin{proof}
First fix an input to the game, $a_j \in \{0,1\}$ with $\sum_{j=1}^T a_j$ even.  Write the (normalized) state $|\psi\rangle$ as
\begin{equation}
    |\psi\rangle = \sum_{\vec{\sigma} \in \{0,1\}^{2N}} c_{\vec{\sigma}} | \sigma_1 \sigma_2\ldots \sigma_{2N} \rangle
\end{equation}
in the computational basis. After Step 2, we have
\begin{equation}
|\psi'\rangle = \sum_{\vec{\sigma}\in\{0,1\}^{2N}} c_{\vec{\sigma}}i^{\sum_{j=1}^T a_j w_j(\vec{\sigma})}|\sigma_1 \sigma_2 \ldots \sigma_{2N}\rangle,
\end{equation}
where to fix the correct branch of $\hat{W}_j^{a_j/2}$, we define $w_j(\vec{\sigma}) \in \{0,1\}$ by
\begin{equation}
w_j(\vec{\sigma}) \equiv \sum_{b\in\Gamma_j} \sigma_b \, \mathrm{mod} \, 2.
\end{equation}
Before performing Step 3, it will be useful to work in the $\hat{X}$ basis at every bond and write $|y_b\rangle= |\hat{X}_b=(-1)^{y_b}\rangle$. This yields
\begin{align}
 |\psi'\rangle = 
\sum_{\vec{y}\in\{0,1\}^{2N}} \left(\frac{1}{2^{N}}\sum_{\vec{\sigma}\in\{0,1\}^{2N}} e^{i(\pi/2)\sum_{j=1}^T a_j w_j(\vec{\sigma})} e^{i\pi \sum_{j=1}^{2N} \sigma_j y_j} c_{\vec{\sigma}}\right) |y_1 y_2 \ldots y_{2N}\rangle.
\end{align}
Then the probability of winning the game with the state $|\psi\rangle$ given the input $\vec{a}$ is given by
\begin{equation}
p(|\psi\rangle,\vec{a}) = \sum_{\substack{\vec{y} \in \{0,1\}^{2N}\\ \sum_{b \in \widetilde{\Gamma}} y_{b} \equiv r \, \mathrm{mod} \,2}} \frac{1}{2^{2N}}\left| \sum_{\vec{\sigma}\in\{0,1\}^{2N}} e^{i(\pi/2)\sum_{j=1}^T a_jw_j(\vec{\sigma})} e^{i\pi \sum_{j=1}^{2N} \sigma_j y_j} c_{\vec{\sigma}}\right|^2
\end{equation}
where $r = \sum_{j=1}^T a_j/2$. We can write this as
\begin{align}
p(|\psi\rangle,\vec{a}) = \frac{1}{2^{2N}}  \sum_{\vec{\sigma},\vec{\sigma}'\in\{0,1\}^{2N}}c_{\vec{\sigma}}c_{\vec{\sigma}'}^*e^{i(\pi/2)\sum_{j=1}^T a_j (w_j(\vec{\sigma}) - w_j(\vec{\sigma}'))} \sum_{\substack{\vec{y} \in \{0,1\}^{2N}\\ \sum_{b\in \widetilde{\Gamma}} y_{b} \equiv r \, \mathrm{mod} \,2}} e^{i\pi \sum_{j=1}^{2N} (\sigma_j-\sigma_j')y_j}. 
\end{align}
Then
\begin{align}
\sum_{\substack{\vec{y} \in \{0,1\}^{2N}\\ \sum_{b \in \widetilde{\Gamma}} y_{b} \equiv r \, \mathrm{mod} \,2}} e^{i\pi \sum_{j=1}^{2N} (\sigma_j-\sigma_j')y_j} &= \sum_{\vec{y} \in \{0,1\}^{2N-|\widetilde{\Gamma}|}} e^{i\pi \sum_{b \not{\in} \widetilde{\Gamma}} (\sigma_b-\sigma_b')y_b}\sum_{\substack{\vec{y} \in \{0,1\}^{|\widetilde{\Gamma}|}\\ \sum_{b\in \widetilde{\Gamma}} y_{b} \equiv r \, \mathrm{mod} \,2}} e^{i\pi \sum_{b \in \widetilde{\Gamma}} (\sigma_b-\sigma_b')y_b}.
\end{align}

By the identity
\begin{equation}
\label{identity2}
\sum_{\{\vec{y}:\sum_{j=1}^M y_j \equiv r \,\mathrm{mod} \,2\}} \prod_{j=1}^M z_j^{y_j} = \frac{1}{2} \left( \prod_{j=1}^M (1+z_j) + (-1)^r  \prod_{j=1}^M (1-z_j)\right),
\end{equation}
we have
\begin{align}
\sum_{\vec{y} \in \{0,1\}^{2N-|\widetilde{\Gamma}|}} e^{i\pi \sum_{b \not{\in}\widetilde{\Gamma}} (\sigma_b-\sigma'_b)y_b} = 2^{2N-|\widetilde{\Gamma}|}\prod_{b \not{\in} \widetilde{\Gamma}} \delta_{\sigma_b,\sigma'_b}
\end{align}
and
\begin{equation}
\sum_{\substack{\vec{y} \in \{0,1\}^{|\widetilde{\Gamma}|}\\ \sum_{b\in \widetilde{\Gamma}} y_{b} \equiv r \, \mathrm{mod} \,2}} e^{i\pi \sum_{b \in \widetilde{\Gamma}} (\sigma_b-\sigma_b')y_b} = 2^{|\widetilde{\Gamma}|-1} \left( \prod_{b \in \widetilde{\Gamma}} \delta_{\sigma_b,\sigma'_b} +  (-1)^r  \prod_{b \in \widetilde{\Gamma}} \delta_{\sigma_b,1-\sigma'_b} \right).
\end{equation}
Thus
\begin{equation}
\nonumber \sum_{\substack{\vec{y} \in \{0,1\}^{2N}\\ \sum_{b \in \widetilde{\Gamma}} y_{b} \equiv r \, \mathrm{mod} \,2}} e^{i\pi \sum_{j=1}^{2N} (\sigma_j-\sigma_j')y_j} = 2^{2N-1}\left( \prod_{b} \delta_{\sigma_b,\sigma'_b} +  (-1)^r  \prod_{b \not{\in} \widetilde{\Gamma}} \delta_{\sigma_b,\sigma'_b} \prod_{b\in \widetilde{\Gamma}}  \delta_{\sigma_b,1-\sigma'_b} \right),
\end{equation}
which implies that
\begin{equation}
p(|\psi\rangle,\vec{a}) = \frac{1}{2}\left(1+ (-1)^r \sum_{\vec{\sigma},\vec{\sigma}'\in\{0,1\}^{2N}}c_{\vec{\sigma}}c_{\vec{\sigma}'}^*e^{i(\pi/2)\sum_{j=1}^{T} a_j (w_j(\vec{\sigma})-w_j(\vec{\sigma}'))} 
 \prod_{b \not{\in} \widetilde{\Gamma}} \delta_{\sigma_b,\sigma'_b} \prod_{b \in \widetilde{\Gamma}}  \delta_{\sigma_b,1-\sigma'_b}
\right). 
\end{equation}
In the second term, we note that since $w_j \in \{0,1\}$ by definition,
\begin{equation}
w_j(\vec{\sigma}') = \left(\sum_{b\in \Gamma_j} \sigma_b  + 1 - 2\sigma_{t_j}\right) \, \mathrm{mod} \, 2 = 1-w_j(\vec{\sigma}),
\end{equation}
so that
\begin{equation}
w_j(\vec{\sigma}) - w_j(\vec{\sigma'}) = 1-2w_j(\vec{\sigma}).
\end{equation}
Thus 
\begin{equation}
e^{i(\pi/2)\sum_{j=1}^{T} a_j (w_j(\vec{\sigma})-w_j(\vec{\sigma}'))} = (-1)^{r+ \sum_{j=1}^T a_j w_j(\vec{\sigma})},
\end{equation}
which finally yields
\begin{equation}
p(|\psi\rangle,\vec{a}) = \frac{1}{2}\left(1+\sum_{\vec{\sigma} \in \{0,1\}^{2N}}(-1)^{\sum_{j=1}^T a_j w_j(\vec{\sigma})} c_{\vec{\sigma}_{\widetilde{\Gamma}},\vec{\sigma}_{\widetilde{\Gamma}^c}}c_{\vec{1}-\vec{\sigma}_{\widetilde{\Gamma}},\vec{\sigma}_{\widetilde{\Gamma}^c}}^* \right).
\end{equation}
We now average uniformly over allowed inputs $a_j$ to yield the quantum winning probability
\begin{align}
\nonumber p_{\mathrm{qu}}(|\psi\rangle) = \frac{1}{2^{T-1}} \sum_{\substack{\vec{a} \in \{0,1\}^{T} \\ \sum_{j=1}^{T} a_j \, \mathrm{even}}} p(|\psi\rangle,\vec{a}) =  \frac{1}{2} +  \frac{1}{2^T}\sum_{\vec{\sigma} \in \{0,1\}^{2N}}c_{\vec{\sigma}_{\widetilde{\Gamma}},\vec{\sigma}_{\widetilde{\Gamma}^c}}c_{\vec{1}-\vec{\sigma}_{\widetilde{\Gamma}},\vec{\sigma}_{\widetilde{\Gamma}^c}}^* \sum_{\substack{\vec{a} \in \{0,1\}^T \\ \sum_{j=1}^T a_j \, \mathrm{even}}} (-1)^{\sum_{j=1}^T a_j w_j(\vec{\sigma})}.
\end{align}
By Eq. \eqref{identity2},
\begin{align}
\sum_{\substack{\vec{a} \in \{0,1\}^T\\ \sum_{j=1}^T a_{j} \, \mathrm{even}}} (-1)^{\sum_{j=1}^T a_j w_j(\vec{\sigma})} = 2^{T-1} \left( \prod_{j=1}^T \delta_{w_j(\vec{\sigma}),0} + \prod_{j=1}^T \delta_{w_j(\vec{\sigma}),1} \right).
\end{align}
Then, imposing the Kronecker delta constraints above, we find that
\begin{align}
p_{\mathrm{qu}}(|\psi\rangle) =  \frac{1}{2} +  \frac{1}{2}\sum_{\{\vec{\sigma} \in \{0,1\}^{2N}: w_j(\vec{\sigma}) = 0\}} \left(c_{\vec{\sigma}_{\widetilde{\Gamma}},\vec{\sigma}_{\widetilde{\Gamma}^c}}c_{\vec{1}-\vec{\sigma}_{\widetilde{\Gamma}},\vec{\sigma}_{\widetilde{\Gamma}^c}}^* + c_{\vec{1}-\vec{\sigma}_{\widetilde{\Gamma}},\vec{\sigma}_{\widetilde{\Gamma}^c}}c_{\vec{\sigma}_{\widetilde{\Gamma}},\vec{\sigma}_{\widetilde{\Gamma}^c}}^* \right).
\end{align}
The physical interpretation of this result is a summation over fidelities to ``topological cat states''. To be explicit, we define the $2^{2N}$ pairs of GHZ-like states
\begin{equation}
|\varphi^{\pm}_{\widetilde{\Gamma}}(\vec{\sigma})\rangle = \frac{1}{\sqrt{2}} \left( \bigotimes_{b \in \widetilde{\Gamma}} |\sigma_b\rangle   \pm \bigotimes_{b \in \widetilde{\Gamma}} |1-\sigma_b\rangle\right) \bigotimes_{b \in \widetilde{\Gamma}^c} |\sigma_b\rangle,
\end{equation}
Defining the dual Wilson line $\hat{V}_{\widetilde{\Gamma}} = \prod_{b\in\widetilde{\Gamma}} X_b$ and writing $|\vec{\sigma}\rangle =  \bigotimes_{b} |\sigma_b\rangle$, these states can be written as
\begin{align}
|\varphi^{\pm}_{\widetilde{\Gamma}}(\vec{\sigma})\rangle = \frac{1}{\sqrt{2}}(1 \pm \hat{V}_{\widetilde{\Gamma}})|\vec{\sigma}\rangle.
\end{align}

In terms of these states, the quantum winning probability is given by
\begin{equation}
\nonumber p_{\mathrm{qu}}(|\psi\rangle) =  \frac{1}{2} \left(1+  \sum_{\{\vec{\sigma} \in \{0,1\}^{2N}: \hat{W}_j|\vec{\sigma}\rangle = |\vec{\sigma}\rangle\}} |\langle \psi |\varphi_{\widetilde{\Gamma}}^{+}(\vec{\sigma})\rangle|^2 - |\langle \psi |\varphi^{-}_{\widetilde{\Gamma}}(\vec{\sigma})\rangle|^2 \right).
\end{equation}
\end{proof}




\end{document}